\documentclass[12pt, fleqn]{article}
\usepackage[cp1251]{inputenc}
\usepackage{latexsym,amsfonts,amssymb}
\usepackage{graphicx}

\usepackage{amsbsy}
\usepackage{amsmath}
\usepackage{epsf}
\usepackage{cite}
\usepackage{hyperref,color}

\newtheorem{theo}{Theorem}
\newtheorem{remark}{Remark}

\newcommand{\bt}{\begin{theo}}
\newcommand{\et}{\end{theo}}
\newcommand{\bd}{\begin{displaymath}}
\newcommand{\ed}{\end{displaymath}}

\newcommand{\lf}{\left}
\newcommand{\rg}{\right}

\newcommand{\be} {\begin{equation}}
\newcommand{\ee} {\end{equation}}
\newcommand{\ba}{\begin{array}}
\newcommand{\ea} {\end{array}}
\newcommand{\bea}{\begin{eqnarray}}
\newcommand{\eea} {\end{eqnarray}}

\newcommand{\p} {\partial}

\newcommand{\lbd} {\lambda}

\begin{document}

\begin{center}
{\Large \bf
 Nonlinear systems of PDEs \\
 admitting  infinite-dimensional Lie algebras \\
 and their connection with Ricci flows. \\ II: The two-dimensional space case.  }

\medskip

{\bf Roman Cherniha $^{a,b,}$\footnote{\small  Corresponding author.
E-mail: r.m.cherniha@gmail.com; roman.cherniha1@nottingham.ac.uk}
and John R. King$^{a}$ }

 $^{a}$ \quad School of Mathematical Sciences, University of Nottingham,\\
  University Park, Nottingham NG7 2RD, UK \\
$^{b}$ \quad National University of Kyiv-Mohyla Academy,\\
2, Skovoroda Street, Kyiv  04070, Ukraine.\\

\end{center}

\begin{abstract}
The work is a natural continuation of that published
in {\it Stud Appl Math. 2024; 153:e12737}.  All possible
two-components evolutions systems of (1+2)-dimensional second-order
PDEs admitting an infinite-dimensional Lie algebra are constructed.
It is shown that a natural generalisation of this Lie algebra to the
higher-dimensional case does not lead to a more general result
because the infinite-dimensional symmetry is broken.  The recently
derived system, which is related to Ricci flows, is identified as a
very particular case among the evolution systems obtained.  All
possible stationary solutions of this system in the
radially symmetric case are constructed using the surprisingly rich
Lie algebra  of the reduced system of ODEs. Moreover, it is proved
that this Lie algebra  is reducible to the fifteen-dimensional
algebra of the simplest system of two second-order ODEs.  Several
time-dependent exact solutions in the radially symmetric case are
constructed as well. It is shown that the solutions obtained are
bounded and smooth provided  arbitrary parameters are
correctly specified.
\end{abstract}


\section{\bf  Introduction} \label{sec:1}

In this work, we are concerned with  evolutionary (1+2)-dimensional
systems admitting the infinite-dimensional Lie algebra generated by
the Lie symmetry
\begin{equation}\label{1-1}
X_{\infty}=\xi^1\partial_{x_1}+ \xi^2\partial_{x_2}+ \xi^1_{x_1}I,
\end{equation}
where the functions $\xi^1(x_1,x_2)$ and $\xi^2(x_1,x_2)$ form   an
arbitrary  solution of the Cauchy-Riemann system
\begin{equation}\label{1-2}  \xi^1_{x_1}=\xi^2_{x_2}, \ \xi^1_{x_2}=-\xi^2_{x_1} \end{equation}
 and $I$ is a so-called unit operator to be specified in the next section.

 One may consider this study as a natural continuation of our previous
 work \cite{ch-ki-24}, in which the systems were studied in the
  case of a single space variable, our focus being on the identification of
   two-dimensional (in space) evolution
   systems possessing an  infinite-dimensional Lie symmetry.
   Actually, our work \cite{ch-ki-24} was inspired by a recent
   result, which says that
  a well-known 1D PDE system, related to the Ricci flows \cite{an-is-knopf-2011,
  an-is-knopf-2015},
  is invariant with respect to an infinite-dimensional Lie algebra \cite{lo-dimas-bo-2023}.
  Here we show that a generalisation of this system to 2D space  case still
  admits the infinite-dimensional Lie algebra generated by the operator $X_{\infty}$;
   however, there is no such
  a  Lie symmetry in the 3D  and higher cases.


There are limited known examples of evolutionary systems admitting
an infinite-dimensional Lie algebra generated by  (\ref{1-1}). For
example, a two-component system, which is a generalisation of the
well-known fast diffusion equation
\begin{equation}\label{1-3} u_{t}=\nabla \cdot(u^{-1}\nabla u), \end{equation}
reads as
\begin{equation}\label{2-1} \ba{l} \medskip
u_t = d_1\nabla \cdot(u^{-1}\nabla u) + u f(u/v), \\
  v_t =    d_2\nabla \cdot(v^{-1}\nabla v) +v g(u/v),
  \ea \end{equation}
where $\nabla=(\frac{\p}{\p x_1}, \frac{\p}{\p x_2})$,  the dot
denotes the  scalar product and $f$ and $g$ are arbitrary smooth
functions. Any system of the form (\ref{2-1})  admits the
infinite-dimensional Lie algebra   (\ref{1-1}) with $I= -2(u\p_u+
v\p_{v})$ \cite{ch-ki-06}. Notably, the infinite-dimensional Lie
algebra of invariance of equation (\ref{1-3}) was identified for the
first time  in \cite{na-70} and   was much later  rediscovered by other
authors (see, e.g., \cite{cim-con-2006}) because this equation  describes  Ricci flow on a
two-dimensional manifold \cite{das-hamilton-2004}. Scalar
(1+2)-dimensional equations with reactive and convective terms
admitting infinite-dimensional symmetry are identified in
\cite{ch-se-pr-2020}. Another example can be identified from the
recent study \cite{lo-dimas-bo-2023} (see (46) with $n=2$ therein),
which leads to the Lie symmetry (\ref{1-1}) with $I= u\p_u.$ We
consider this example in detail in Section 4 below.

Taking into account the examples presented above, one may suggest
the following form for the unit operator
\begin{equation}\label{2-3} I = \lbd_1 u\p_u+ \lbd_2
v\p_{v}, \end{equation}
 where $\lbd_1$ and $\lbd_2$ are given constants.
It turns out that the above  constants  are reducible to
$\lbd_1=\lbd_2=1$  by appropriate local transformations. Indeed,
assuming  $\lbd_1\lbd_2\not=0$, one easily calculates that
(\ref{2-3}) reduces to
\begin{equation}\label{2-4} I =  u^*\p_u^*+
v^*\p_{v^*} \end{equation} by the transformation $u^*=u^{1/\lbd_1},
\  v^*=v^{1/\lbd_2}$. If  $\lbd_1\lbd_2=0$ and, for example,
$\lbd_1\not=0$ then the  transformation $u^*=u^{1/\lbd_1}, \
v^*=u^{1/\lbd_1}v$ again leads to the form  (\ref{2-4}).  Thus, we
assume $\lbd_1=\lbd_2=1$  in what follows without losing 
generality. In particular, applying the transformation
\[ u = u^{*-2}, \ v = v^{*-2} \]
to the fast diffusion system (\ref{1-1}), one readily obtains the
system
\begin{equation}\label{2-5} \ba{l} \medskip
u^*_t = d_1\Big(u^{*2}\Delta u^* - u^*|\nabla u^*|^2\Big) + u^* f^*(u^*/v^*), \\
  v^*_t =    d_2\Big(v^{*2}\Delta v^* - v^*|\nabla v^*|^2\Big) + v^* g^*(u^*/v^*),
  \ea \end{equation}
where $f^*=-\frac{1}{2} f(v^{*2}/u^{*2})$ and $g^*=-\frac{1}{2}
g(v^{*2}/u^{*2})$. The nonlinear system  (\ref{2-5})  admits the
Lie symmetry (\ref{1-1}) with (\ref{2-4}).

The main aim of this study is to describe in an explicit form all
two-component evolutionary systems of the form (\ref{3-1}) below
admitting the infinite-dimensional  Lie algebra generated by the Lie
symmetry (\ref{1-1}). This is an important problem in Lie symmetry
analysis of PDEs because
 classical mathematical models
arising in physics (the heat and wave equations, Navier-Stokes
equations, the Maxwell  equations, etc.) admit non-trivial Lie
algebras of invariance,  which often reflect the fundamental laws of
physics. Thus, constructing  generalizations of these  equations
(systems of PDEs)  that   preserve the relevant   symmetries
is one of the most important problems in Lie symmetry analysis. Many
such  problems were solved  in past: in particular,   wide classes
of nonlinear equations (nonlinear systems of PDEs) admitting Galilei
algebra, Poincare algebra, Euclid algebra, conformal algebra,
Schr\"odinger algebra etc.  have been constructed (see an extensive
list of references in \cite{ch-ki-24}). However, there are not many
papers devoted to solving such  problems with respect to
infinite-dimensional Lie algebras,
 though some  examples can be found in \cite{fss},  \cite{ch-2001},\cite{ch-he-2004}.
 A  most recent example was derived in \cite{ch-ki-24}.

 The following sections are  organized as follows. In Section~\ref{sec:3}, a wide
class of two-component evolution systems  consisting of
(1+2)-dimensional PDEs  is constructed admitting the
infinite-dimensional Lie algebra generated by the operator
(\ref{1-1}). Notably, there are no linear evolution systems
belonging to the class. In Section~\ref{sec:4}, it is identified
that the nonlinear system with two space variables  derived in
\cite{lo-dimas-bo-2023} can be obtained from the above mentioned
class as a very particular case. We also discuss Lie symmetries of
the two-component system and that with two additional equations,
which form the Cauchy-Riemann system in an unusual form. It is shown
that the conditional symmetry concept \cite{fss} is helpful  in
order to explain how Lie symmetry depends on the  number (two, three and
four)  of PDEs in the system in question. In Section~\ref{sec:5},
radially symmetric solutions in closed form are constructed for the
two-component nonlinear system related to the Ricci flows. In
particular, all possible  stationary solutions  are constructed
using a surprisingly rich Lie symmetry of the reduced system of
ODEs. Several time-dependent solutions are derived as well. Finally,
we discuss the results obtained and present some conclusions  in
Section~\ref{sec:6}. In particular, it is shown that a direct
generalisation of the Lie symmetry (\ref{1-1}) to three and more
space variables does not lead to evolution systems admitting
infinite-dimensional Lie algebras.

\section{\bf Main result about evolution systems with infinite-dimensional Lie symmetries } \label{sec:3}

We consider two-component evolution systems involving   derivatives
of the first order  in time and up to the second  order of space
variables  of the general form
\begin{equation}\label{3-1} \ba{l} \medskip
u_t= F(u, v,  \underset{1}{u}, \underset{1}{v}, \underset{11}{u}, \underset{11}{v}),\\
v_t= G(u, v, \underset{1}{u}, \underset{1}{v}, \underset{11}{u},
\underset{11}{v} ), \ea \ee where $F$ and $G$ are arbitrary
differentiable functions of their arguments and the following
notations are used: $\underset{1}{u}= (u_{x_1}, u_{x_2}), \
\underset{11}{u}= (u_{x_1x_1}, u_{x_1x_2}, u_{x_2x_2})$,
$\underset{1}{v}= (v_{x_1}, v_{x_2}), \ \underset{11}{v}=
(v_{x_1x_1}, v_{x_1x_2}, v_{x_2x_2})$.

In what follows we assume that both of the above equations are of
second-order and the system is not reducible to one involving a
first-order PDE.


\begin{theo} \label{th1}
An evolution system of the from (\ref{3-1}) admits   the
infinite-dimensional  Lie algebra generated by the Lie symmetry
\begin{equation}\label{3-2}
X_{\infty}=\xi^1\partial_{x_1}+ \xi^2\partial_{x_2}+
\xi^1_{x_1}(u\p_u+ v\p_{v}),
\end{equation}
 if and only if this system possesses the form
\begin{equation}\label{3-3} \ba{l}
u_t= u f\Big(I_0, \ I_1,\ I_{11}, \  I_{22}, \ J \Big),\\
v_t= v g\Big( I_0, \ I_1,\ I_{11}, \  I_{22}, \ J \Big), \ea
\end{equation} where $f$ and $g$ are arbitrary differentiable
functions of their arguments, while
\begin{equation}\label{3-3a} \ba{l}  I_0=\frac{v}{u}, \
I_1=\Big|\frac{v}{u}\nabla u - \nabla v\Big|^2, \ I_{11}= u\Delta u-
|\nabla u|^2,  \ I_{22}= v\Delta v- |\nabla v|^2, \\
   J= \Big( v(u_{x_1x_1}-u_{x_2x_2})- u(v_{x_1x_1}-v_{x_2x_2})\Big)^2 +4\Big(vu_{x_1x_2}-uv_{x_1x_2}\Big)^2.
   \ea
\end{equation}
\end{theo}

\textbf{Proof. } This  is based on the classical  invariance
criterium  that was worked out by Sophus Lie and can be found in any
book devoted to the Lie symmetry analysis  (see, e.g., \cite{olv-93,
bl-anco-10,ch-se-pl-book}). The main difficulty in application of
the criteria to the evolutionary  system (\ref{3-1}) follows from
the structure of the Lie symmetry (\ref{3-2}) because the latter
involves two otherwise  arbitrary functions satisfying  the
Cauchy-Riemann system (\ref{1-2}).

First of all, we write down the criteria for (\ref{3-1}):
\begin{equation}\label{3-11}\begin{array}{l}
\mbox{\raisebox{-1.6ex}{$\stackrel{\displaystyle  
X}{\scriptstyle 2}$}}\, \lf(u_{t}- F(u, v,  \underset{1}{u}, \underset{1}{v}, \underset{11}{u}, \underset{11}{v})\rg)\Big\vert_{{\cal{M}}}  
  =0,
\\  \mbox{\raisebox{-1.6ex}{$\stackrel{\displaystyle  
X}{\scriptstyle 2}$}}\, \lf(v_{t}- G(u, v,  \underset{1}{u}, \underset{1}{v}, \underset{11}{u}, \underset{11}{v})\rg)\Big\vert_{{\cal{M}}}  
  =0, \end{array}\end{equation} where the  manifold
     \[{\cal{M}}=
\lf\{u_{t}- F(u, v,  \underset{1}{u}, \underset{1}{v},
\underset{11}{u}, \underset{11}{v})=0, \  v_{t}- F(u, v,
\underset{1}{u}, \underset{1}{v}, \underset{11}{u},
\underset{11}{v})=0 \rg\}\] is considered in the prolonged space of
independent variables
\[
t, \, x_1,\, x_2, \,u, \,v,\, u_t, \,v_t, \, \underset{1}{u}, \
\underset{1}{v},\, u_{tt}, \,v_{tt},  \,u_{tx_1}, \,v_{tx_1},
\,u_{tx_2}, \,v_{tx_2}, \, \underset{11}{u}, \ \underset{11}{v}.\]

The operator $\mbox{\raisebox{-1.6ex}{$\stackrel{\displaystyle  
X}{\scriptstyle 2}$}} $ is the second prolongation of the operator
$X$ and has the structure
 \begin{equation}\label{3-12}\begin{array}{l} \medskip
\mbox{\raisebox{-1.6ex}{$\stackrel{\displaystyle  
X}{\scriptstyle 2}$}} =X+\rho^1_t\frac{\p}{\p
u_t}+\rho^2_t\frac{\p}{\p v_t}+\rho^1_{x_i}\frac{\p}{\p u_{x_i}}+
\rho^2_{x_i}\frac{\p}{\p v_{x_i}}+ \\
\qquad   +\sigma^1_{tt}\frac{\p}{\p
u_{tt}}+\sigma^2_{tt}\frac{\p}{\p v_{tt}} + \sigma^1_{t
x_i}\frac{\p}{\p u_{tx_i}}+\sigma^2_{tx_i}\frac{\p}{\p v_{tx_i}}
+\sigma^1_{x_ix_j}\frac{\p}{\p
u_{x_ix_j}}+\sigma^2_{x_ix_j}\frac{\p}{\p v_{x_ix_j}},
\end{array}\end{equation} where  summation over the repeated
indices $i=1,2$ and $j=1,2$ is assumed and $i\leqq j$.
 The
coefficients $\rho^k$ and $\sigma^k$ ($k=1,2$) with
subscripts are defined by the coefficients of the operator $X$ via
well-known formulae. Note that we do not need to calculate
$\sigma^k_{tt}$ and $\sigma^k_{t x_i}$ because the class of PDEs
(\ref{3-1}) does not contain
  second-order time derivatives or  second-order mixed derivatives involving the time variable.
  As a result,  the coefficients
  can be determined explicitly  in the form:
 \begin{equation}\label{3-13}\begin{array}{l} \medskip
\rho^1_t= \xi^1_{x_1}u_t, \ \rho^2_t= \xi^1_{x_1}v_t, \\
 \medskip
 \rho^1_{x_1}= \xi^1_{x_1x_1}u + \xi^1_{x_2}u_{x_2}, \
\rho^1_{x_2}= \xi^1_{x_1x_2}u - \xi^1_{x_2}u_{x_1}, \\
\medskip
 \rho^2_{x_1}= \xi^1_{x_1x_1}v + \xi^1_{x_2}v_{x_2}, \
\rho^2_{x_2}= \xi^1_{x_1x_2}v - \xi^1_{x_2}v_{x_1}, \\
 \medskip
\sigma^1_{x_1x_2}= \xi^1_{x_1x_1x_2} u - \xi^1_{x_1}u_{x_1x_2} + \xi^1_{x_2}(u_{x_2x_2}-u_{x_1x_1}), \\
 \medskip
\sigma^2_{x_1x_2}= \xi^1_{x_1x_1x_2} v - \xi^1_{x_1}v_{x_1x_2} + \xi^1_{x_2}(v_{x_2x_2}-v_{x_1x_1}), \\
 \medskip
 \sigma^1_{x_1x_1}= \xi^1_{x_1x_1x_1} u+ \xi^1_{x_1x_1}u_{x_1} + \xi^1_{x_1x_2}u_{x_2}- \xi^1_{x_1}u_{x_1x_1} + 2\xi^1_{x_2}u_{x_1x_2}, \\
 \medskip
 \sigma^2_{x_1x_1}= \xi^1_{x_1x_1x_1} v+ \xi^1_{x_1x_1}v_{x_1} + \xi^1_{x_1x_2}v_{x_2}- \xi^1_{x_1}v_{x_1x_1} + 2\xi^1_{x_2}v_{x_1x_2}, \\
 \medskip
 \sigma^1_{x_2x_2}= \xi^1_{x_1x_2x_2} u+ \xi^1_{x_1x_1}u_{x_1} + \xi^1_{x_1x_2}u_{x_2}- 2\xi^1_{x_2}u_{x_1x_2}  -\xi^1_{x_1}u_{x_2x_2}, \\
 \medskip
  \sigma^2_{x_2x_2}= \xi^1_{x_1x_2x_2} v+ \xi^1_{x_1x_1}v_{x_1} + \xi^1_{x_1x_2}v_{x_2}- 2\xi^1_{x_2}v_{x_1x_2}  -\xi^1_{x_1}v_{x_2x_2}.
\ea\end{equation}

  At the next stage, one should insert the explicit expressions from (\ref{3-13})  into (\ref{3-12}) and  apply the second prolongation of
$X$ to the invariance criteria  (\ref{3-11}). Note that  the
manifold ${\cal{M}}$ can be taken into account by excluding the time
derivatives using the system in question, i.e. (\ref{3-1}). As a
result, two very  cumbersome  expressions are obtained (they are
omitted here) and each of them must vanish for arbitrary solutions
of the  Cauchy-Riemann system (\ref{1-2}).  Obviously, the latter
possesses an infinite number of linearly independent solutions,
therefore both expressions obtained can be split  with respect to
the different derivatives of the function  $\xi^1(x_1,x_2)$  (one
may note that the function $\xi^2(x_1,x_2)$ is eliminated using
(\ref{1-2})). Thus, exactly six first-order PDEs  are obtained to
find the function $F$ and the same equations are obtained for  $G$.
At the final stage, one needs to solve the six-component system

 \begin{equation}\label{3-16}\begin{array}{l} \medskip
u \frac{\p F}{\p u}+ v \frac{\p F}{\p v} - u_{x_1x_1}\frac{\p F}{\p u_{x_1x_1}}- u_{x_2x_2}\frac{\p F}{\p u_{x_2x_2}}- u_{x_1x_2}\frac{\p F}{\p u_{x_1x_2}} \\
\medskip
\qquad - v_{x_1x_1}\frac{\p F}{\p v_{x_1x_1}} - v_{x_2x_2}\frac{\p F}{\p u_{x_2x_2}}- v_{x_1x_2}\frac{\p F}{\p u_{x_1x_2}} =F,\\
\medskip
u_{x_2}\frac{\p F}{\p u_{x_1}}-u_{x_1}\frac{\p F}{\p u_{x_2}}+
v_{x_2}\frac{\p F}{\p v_{x_1}}-v_{x_1}\frac{\p F}{\p v_{x_2}}
+2u_{x_1x_2}\Big(\frac{\p F}{\p u_{x_1x_1}}-\frac{\p F}{\p u_{x_2x_2}}\Big)\\
\medskip
+2v_{x_1x_2}\Big(\frac{\p F}{\p v_{x_1x_1}}-\frac{\p F}{\p
v_{x_2x_2}}\Big)+ (u_{x_2x_2}-u_{x_1x_1})\frac{\p F}{\p u_{x_1x_2}}
+ (v_{x_2x_2}-v_{x_1x_1})\frac{\p F}{\p v_{x_1x_2}}=0,
\\
\medskip
u\frac{\p F}{\p u_{x_1}}+v_{x_1}\frac{\p F}{\p v_{x_1}}
+u_{x_1}\Big(\frac{\p F}{\p u_{x_1x_1}}+\frac{\p F}{\p
u_{x_2x_2}}\Big)
+v_{x_1}\Big(\frac{\p F}{\p v_{x_1x_1}}+\frac{\p F}{\p v_{x_2x_2}}\Big)=0, \\
\medskip
u\frac{\p F}{\p u_{x_2}}+v_{x_1}\frac{\p F}{\p v_{x_2}}
+u_{x_2}\Big(\frac{\p F}{\p u_{x_1x_1}}+\frac{\p F}{\p
u_{x_2x_2}}\Big)
+v_{x_2}\Big(\frac{\p F}{\p v_{x_1x_1}}+\frac{\p F}{\p v_{x_2x_2}}\Big)=0, \\
\medskip
u\Big(\frac{\p F}{\p u_{x_1x_1}}-\frac{\p F}{\p u_{x_2x_2}}\Big)+v\Big(\frac{\p F}{\p v_{x_1x_1}}-\frac{\p F}{\p v_{x_2x_2}}\Big)=0, \\
\medskip
u\frac{\p F}{\p u_{x_1x_2}} + v\frac{\p F}{\p v_{x_1x_2}}=0.
\ea\end{equation}
 (\ref{3-16}) is an overdetermined system  of the linear first-order PDEs for the function $F$.
 Each equation can be solved using
the well-known method of characteristics. However, a nontrivial
technical problem occurs because, for example, solving the first
equation from  (\ref{3-16}), one obtains its general solution as an
arbitrary functions of new variables. Each new variable is a known
function of the initial variables. So, in order to solve the next
equation from    (\ref{3-16}), one needs to write down this equation
in new variables and so on. This process in nontrivial because the new successive 
variables have more complicated structure. In order to avoid ugly
expressions, it is necessary to find the simplest structure of new
variables on each step. 

Finally, the general solution of the overdetermined system
(\ref{3-16}) was found in the form
\[ F = u f\Big(I_0, \ I_1,\ I_{11}, \  I_{22}, \ J \Big), \]
where $f$ is an arbitrary differentiable of five arguments that have
exactly the same form as in Theorem 1. The function $G$ has the same
structure and can be presented in the form
 \[ G = v g\Big(I_0, \ I_1,\ I_{11}, \  I_{22}, \ J \Big). \]

The proof is completed.

  $\blacksquare$

\begin{remark}
The expressions $I_0$ and $ I_1$  are absolute differential
invariants of  order zero and one, respectively, while $ I_{11}, \
I_{22}$ and $ J $ are absolute differential invariants of the second
order. Because  absolute differential invariants are defined only up
to functional dependence, each of the above invariants can be
presented in equivalent forms. For example the first order invariant
$I_1$ can
 be replaced by the symmetric form $ \frac{v}{u}|\nabla u|^2 -2\nabla u\cdot\nabla v+ \frac{u}{v}|\nabla v|^2$.
\end{remark}

It should be noted that a large majority of mathematical models
arising in real-world applications involve the second derivatives
only linearly. Thus, the class of governing  systems for such models
admitting the infinite-dimensional  Lie algebra (\ref{3-2}) reads as
\begin{equation}\label{3-3*} \ba{l}
u_t= u f_{11}(I_0, \ I_1)(u\Delta u- |\nabla u|^2) +u f_{12}(I_0, \
I_1)(v\Delta v-
|\nabla v|^2) +uf_{10}(I_0, \ I_1),\\
v_t= vf_{21}(I_0, \ I_1)(u\Delta u- |\nabla u|^2) +v f_{22}(I_0, \
I_1)(v\Delta v- |\nabla v|^2) +vf_{20}(I_0, \ I_1) , \ea
\end{equation} where $f_{ij}$ are arbitrary smooth
functions of the absolute invariants $I_0$ and $ I_1$. In the  very
particular case, $f_{ii}=d_i, \ i=1,2,$  $f_{12}=f_{21}=0$,
$f_{10}=f(I_0)$ and $f_{10}=g(I_0)$, one immediately arrives at
system (\ref{2-5}), which is equivalent to the fast diffusion system
(\ref{2-1}).  Notably, the class of PDEs (\ref{3-3*}) does not
contain {\it linear}  systems of  PDEs.

\begin{remark}
Because  the Lie symmetry (\ref{3-2}) does not involve the variable
$t$, similar systems  for  wave and telegraph type equations
admitting the same infinite-dimensional  Lie algebra  can
immediately be written down.  For example, in the case of systems
involving wave type equations, the analogue of (\ref{3-3})  reads as
  \begin{equation}\label{3-17} \ba{l}
u_{tt}= u f\Big(I_0, \ I_1,\ I_{11}, \  I_{22}, \ J \Big),\\
v_{tt}= v g\Big( I_0, \ I_1,\ I_{11}, \  I_{22}, \ J \Big). \ea
\end{equation}
\end{remark}


Finally, we highlight the following problem that naturally arises from this study.
A natural generalization of the Lie symmetry    (\ref{1-1})  to the  multidimensional case
reads as
\begin{equation}\label{6-1}
X_{\infty}=\xi^1\partial_{x_1}+\dots + \xi^n\partial_{x_n}+ \xi^1_{x_1}I, \quad n\geqq 3.
\end{equation}
where the functions $\xi^i(x_1,\dots, x_n)$ should satisfy the following generalisation
of the Cauchy-Riemann system
\begin{equation}\label{6-2}  \xi^1_{x_1}=\xi^i_{x_i}, \ \xi^i_{x_j}=-\xi^j_{x_i}, \quad i,j= 1,\dots, n, \ i<j.\end{equation}
It can be easily calculated that  (\ref{6-2}) is an overdetermined system of PDEs that consists of
 $\frac{1}{2}(n-1)(n+2)$ equations.
 Having this in hand,  it can be proved by using differential consequences that all third derivatives
 $\xi^1_{x_ix_jx_k}=0, \dots, \xi^n_{x_ix_jx_k}=0$ $(i,j,k= 1,\dots, n)$.
 So,  the functions $\xi^i(x_1,\dots, x_n)$ are quadratic polynomials, therefore
 a finite-dimensional Lie algebra that is nothing else but the conformal algebra is obtained (absolute differential
 invariants of this algebra for scalar PDEs can be found in \cite{fu-ye-1992}).
 In other words,  the Lie symmetry    (\ref{6-1}) does not generate an infinite-dimensional
Lie algebra. This means that another generalisation of  the Lie
symmetry (\ref{1-1}) would need to  be identified in order to
generalise the results obtained here to multidimensional cases
$n\geqq 3$.


\section{\bf  Systems related  to the  Ricci flow for specific metrics } \label{sec:4}

Recently \cite{lo-dimas-bo-2023},
  Ricci flow on warped product manifolds with a given metric
has been investigated by the Lie symmetry method. As a result,
 a complicated nonlinear system of PDEs was derived.
  In the simplest case of reduction to a single space variable, the
system was studied in order to use its Lie symmetry for constructing
exact solutions. The system
 can be written in the form (see (46) with $n=1$ in
\cite{lo-dimas-bo-2023})
\begin{equation}\label{4-1} \ba{l}
(vw)_t= (1-m)v^2\Big(vw_x\Big)_x + \frac{v}{w}\Big((m-1)v^2 w_x^2 -\mu\Big),\\
v_t= -m\frac{v^2}{w}\Big(vw_x\Big)_x, \ea \end{equation} where $m
=1,2,3...$ and $\mu$ is a real parameter (which can  also  be  a
function of time). Notably,  the nonlinear  system (\ref{4-1}) can
also be derived from  earlier studies
\cite{an-is-knopf-2011,an-is-knopf-2015}. In \cite{ch-ki-24}, it was
shown how  (\ref{4-1}) can be identified as a very particular case
from the general class of two-component systems of PDEs admitting an
infinite-dimensional Lie algebra.

It turns out that one arrives at a much more complicated systems if
the number of space variables is two or more, i.e. $n>1$ (in the
notation of \cite{lo-dimas-bo-2023}). Here we consider the case
$n=2$ in detail (though give some  discussion about the cases $n>2$
in the end of this study). The system of equations (46) with $n=2$
\cite{lo-dimas-bo-2023}  can be presented in the form
\begin{equation}\label{4-2} \ba{l}
ww_t= wv^2 \Delta w + (m-1)v^2|\nabla w|^2 -\mu,\\
v_t= v^2 \Delta v- v|\nabla v|^2-\frac{mv^3}{2w}\Delta w,\\
(v^2 w_{x_1})_{x_1}= (v^2 w_{x_2})_{x_2}, \\
(v^2 w_{x_1})_{x_2}= -(v^2 w_{x_2})_{x_1}.
 \ea \end{equation}
 One notes that this is an overdetermined four-component system of
 PDEs. The first two PDEs are standard evolution equations and their
 structure is similar to that in the one-dimensional  case. The
 last  two equations form the Cauchy-Riemann system  for the
 functions  $A=v^2 w_{x_1}$  and $B=v^2 w_{x_2}.$

Let us show that the first two PDEs from (\ref{4-2}) form a system
that is nothing else but a very particular case of the   general
system (\ref{3-3}): therefore this system indeed admits the
infinite-dimensional  Lie algebra (\ref{3-2}). Indeed, applying the
transformation $u=wv$, the first two equations of (\ref{4-2}) reduce
to the form
\begin{equation}\label{4-3} \ba{l}
u_t= u\Big(v \Delta v- |\nabla v|^2 + (m-1)\Big|\frac{v}{u}\nabla u - \nabla v\Big|^2\Big) \\
\qquad +\frac{2-m}{2}v\Big(v \Delta u -u \Delta v-2\nabla u\cdot\nabla v +\frac{2u}{v}|\nabla v|^2\Big) -\mu \frac{v^2}{u},\\
v_t= v\Big(v \Delta v- |\nabla v|^2 - \frac{mv}{2u} \Big(v \Delta u
-u \Delta v-2\nabla u\cdot\nabla v +\frac{2u}{v}|\nabla
v|^2\Big)\Big).
 \ea \end{equation}
Now it can be noted that the right-hand-sides of  the nonlinear
system (\ref{4-3}) can be expressed via the  absolute differential
invariants $I_0, \ I_1,\ I_{11},$ and   $I_{22}$. As a result, we
obtain the system
\begin{equation}\label{4-4} \ba{l}
u_t= u\Big(I_{22} + (m-1)I_1 -\mu I_0^2 + \frac{2-m}{2}(I_0^2I_{11} -
I_{22}
+I_1)\Big), \\
 v_t= v\Big(I_{22} - \frac{m}{2}(I_0^2I_{11} - I_{22}
+I_1)\Big).\ea \end{equation}

Thus, the evolutionary system (\ref{4-3}) admits the infinite-dimensional
Lie algebra (\ref{3-2}). As a consequence, one may claim that the
first  two PDEs from (\ref{4-2}) admit the infinite-dimensional Lie
algebra independently of the final two equations.

The last  two PDEs from (\ref{4-2}) take the form
\begin{equation}\label{4-5} \ba{l}
v(u_{x_1x_1}-u_{x_2x_2})- u(v_{x_1x_1}-v_{x_2x_2})=0, \\
vu_{x_1x_2}-uv_{x_1x_2}=0 \ea \end{equation}
 after applying the
transformation $u=wv$.  On the other hand, if one considers the
simplest equation created by the differential invariant $J$ (see
Theorem 1), namely $J=0$:
\[\Big( v(u_{x_1x_1}-u_{x_2x_2})- u(v_{x_1x_1}-v_{x_2x_2})\Big)^2
+4\Big(vu_{x_1x_2}-uv_{x_1x_2}\Big)^2 =0, \]
  then immediately  the nonlinear stationary  system
(\ref{4-5}) is obtained (we remind the reader that the functions $u$ and $v$ are real-valued).
 Thus, the entire overdetermined  system  (\ref{4-2})
derived in  \cite{lo-dimas-bo-2023} indeed admits the
infinite-dimensional Lie algebra.

The above observation is related to the notion of conditional
invariance \cite{fss} (see Section 5.7 therein), though, one may
highlight an unusual  situation. Indeed, according to the definition
of conditional invariance, a given PDE (system of PDEs) is
conditionally invariant under an operator  if and only if the
overdetermined system consisting of both   the  given PDE (system of
PDEs) and an additional equation admits this operator. Probably, the
first highly non-trivial example of  conditional invariance was
found in study \cite{fu-se-am-90}, in which the conditional
conformal invariance  of the RD equation
\[ u_t=(e^u u_x)_x +e^{-u} \]
is demonstrated (see also examples of   conditional conformal
invariance for (1+2)-dimensional PDEs in \cite{ch-he-2010}). It
turns out that the maximal  algebra of invariance (MAI) of the
nonlinear overdetermined  system (\ref{4-2}) and the nonlinear
system involving only the first two  equations of (\ref{4-2})  is
the same. The relevant MAI is an infinite-dimensional algebra
generated by the symmetry operators
\begin{equation}\label{4-6}\ba{l}
X_{\infty}=\xi^1\partial_{x_1}+ \xi^2\partial_{x_2}+
\xi^1_{x_1}v\p_{v},\\
P_t= \partial_t, \ D= 2t\partial_t- v\p_{v}+ w\p_{w}. \ea
\end{equation} Thus, the notion of conditional symmetry gives
nothing in this context if one considers the last  two  equations of
(\ref{4-2}) as additional conditions.


However, if one takes the subsystem of (\ref{4-2})
\begin{equation}\label{4-7} \ba{l}
ww_t= wv^2 \Delta w + (m-1)v^2|\nabla w|^2 -\mu,\\
v_t= v^2 \Delta v- v|\nabla v|^2-\frac{mv^3}{2w}\Delta w,\\
(v^2 w_{x_1})_{x_1}= (v^2 w_{x_2})_{x_2}
 \ea \end{equation}
then it can be checked in a straightforward way that the
three-component system (\ref{4-7}) is not invariant under the
operator $X_{\infty}$ from  (\ref{4-6}) but  only conditionally
invariant,
the corresponding
condition is the fourth equation from system (\ref{4-2}).

\section{\bf  Exact solutions of a nonlinear system related  to the  Ricci flow} \label{sec:5}

As  was shown above, system (\ref{4-3}) is related  to the  Ricci
flow and can be identified from the recent study
\cite{lo-dimas-bo-2023}. In contrast to the 1D case that was
analysed earlier, in \cite{ch-ki-24}, the 2D  system (\ref{4-3}) is
much more complicated.
Setting $m=2$ and $\mu=0$ leads to a simplification, we therefore
consider in what follows the nonlinear evolution system
\begin{equation}\label{5-1} \ba{l}
u_t= u\Big(v \Delta v- |\nabla v|^2 + \Big|\frac{v}{u}\nabla u - \nabla v\Big|^2\Big),\\
v_t= v\Big(v \Delta v- |\nabla v|^2 - \frac{v}{u} \Big(v \Delta u -u
\Delta v-2\nabla u \cdot\nabla v +\frac{2u}{v}|\nabla
v|^2\Big)\Big).
 \ea \end{equation}

 From the applicability point of view, exact solutions with radial symmetry are likely to be the
  most interesting in the 2D case.
 In contrast to travelling  wave solutions,   such  solutions are essentially different
 from those in the 1D case
 because radially symmetric solutions  cannot of course  be derived  by the ansatz
  $u(t,x_1,x_2)=u(t,z), \ z= c_1x_1+c_2x_2$ ($c_1$ and $c_2$ are some constants), which reduces (\ref{5-1}) to its 1D version.

  Radially symmetric solutions of (\ref{5-1}) are obtainable via the ansatz
  \begin{equation}\label{5-2}  u=u(t,r), \ v=v(t,r), \  r=\sqrt{x_1^2+x_2^2}
  \end{equation}
that is produced by the rotation operator $x_2\partial_{x_1}-x_1\partial_{x_2}$. The latter is a very particular case
of the Lie symmetry $X_{\infty}$  (\ref{1-1}).
Ansatz (\ref{5-2}) reduces system (\ref{5-1}) to the form

\begin{equation}\label{5-3} \ba{l}
u_{t}= u\Big[v \left(v_{r,r}+\frac{v_{r}}{r}\right)+\frac{v^{2} }{u^{2}}u_{r}^{2}-\frac{2 v }{u}u_{r}v_{r} \Big],
\\
v_{t}=v \Big[2 v \left(v_{r,r}+\frac{v_{r}}{r}\right)-\frac{v^{2} }{u}\left(u_{r,r}+\frac{u_{r}}{r}\right)-3 v_{r}^{2}+\frac{2 v}{u}u_{r}v_{r}\Big].
 \ea \end{equation}

 \begin{theo} \label{th2}
The nonlinear  evolution system  (\ref{5-3}) admits   the five-dimensional MAI
 generated by the Lie symmetries
 \begin{equation}\label{5-4}\ba{l}
P_t= \partial_t, \ I_u=u\p_{u}, D_{00}= 2t\partial_t +r\p_{r}, \ D_{11}= r\p_t+u\p_{u}+ v\p_{v},\\
X=\ln r \Big( r\p_r+u\p_{u}+ (1+\ln^{-1}r)v\p_{v}\Big).
\ea \end{equation}
\end{theo}

\textbf{Sketch of Proof. }
 The system in question does not involve  coefficients in the form of arbitrary  parameters and/or functions.
 Nowadays Lie symmetry of such systems can be calculated using computer algebra packages.
We used Maple in order to derive MAI (\ref{5-4}).

\medskip

Because the Lie symmetry of (\ref{5-3}) is rich,  several
inequivalent ans\"atze can be constructed using the Lie algebra
(\ref{5-4}). Notably, the Lie symmetry $X$ has  unusual form, which
is not related to the Lie symmetry $X_{\infty}$. Here we present two
important examples leading to stationary solutions and exact
solutions with separated variables.

\subsection{Stationary solutions}

Reduction of system (\ref{5-3}) via the Lie symmetry $P_t= \partial_t$
leads to the ODE system

\[
V \left(V_{r,r}+\frac{V_{r}}{r}\right)+\frac{V^{2} \left(U_{r}\right)^{2}}{U^{2}}-\frac{2 V \left(U_{r}\right) \left(V_{r}\right)}{U}=0,
\] \[
2 V \left(V_{r,r}+\frac{V_{r}}{r}\right)-\frac{V^{2} \left(U_{r,r}+\frac{U_{r}}{r}\right)}{U}-3 \left(V_{r}\right)^{2}+\frac{2 V \left(U_{r}\right) \left(V_{r}\right)}{U}=0,
\]
which is equivalent to the system
\begin{equation}\label{5-5} \ba{l}
 V \left(U_{r,r}+\frac{U_{r}}{r}\right)+\frac{2V}{U}U_{r}^2+\frac{3U}{V}V_{r}^2- 6U_{r}V_{r}=0 \\
 \medskip\\
V \left(V_{r,r}+\frac{V_{r}}{r}\right)+\frac{V^{2} }{U^{2}}U_{r}^{2}-\frac{2 V}{U} U_{r}V_{r}=0.
 \ea \end{equation}
 Each pair of the function  $U(r)$ and $V(r)$ is a stationary solution of (\ref{5-4}) provided
 these functions satisfy the ODE system (\ref{5-5}).  Although  the ODE system (\ref{5-5})
  is complicated,
  it  possesses a surprisingly rich Lie symmetry.

  \begin{theo} \label{th3}
The nonlinear  ODE system  (\ref{5-3}) admits   the fifteen-dimensional MAI
 generated by the Lie symmetries
 \begin{equation}\label{5-6}\ba{l}
X_{1}=\p_r, \ X_{2}=r\ln r\p_r,   \ X_{3}=r\frac{U^{2} }{V^{2}}\p_r, \ X_{4}=(\ln U -3 \ln V) r\p_r,\ X_{5}=U\p_{U}, \ X_{6}=V\p_{V}, \\
  X_{7}=  \ln r(U\p_{U}+ V\p_{V}), \ X_{8}=\frac{V^{2}}{ U}\p_U +\frac{V^{3}}{3 U^{2}}\p_V, \
X_{9}=\frac{U^{3}}{ V^{2}}\p_U +\frac{U^{2}}{ V}\p_V, \ X_{10}=\ln r\Big(  \frac{V^{2}}{ U}\p_U +\frac{V^{3}}{3 U^{2}}\p_V \Big),\\
 X_{11}=  (\ln U -3 \ln V)(U\p_{U}+ V\p_{V}), \
X_{12}=(2\ln U -6 \ln V+3)\Big(  \frac{V^{2}}{ U}\p_U +\frac{V^{3}}{3 U^{2}}\p_V \Big), \\
X_{13}=\ln r \Big( 2r\ln r\p_r-(\ln U -3 \ln V)U\p_{U}- (\ln U -3 \ln V+1)V\p_{V}\Big), \\
X_{14}=(\ln U -3 \ln V) \Big( 2r\ln r\p_r-(\ln U -3 \ln V-3)U\p_{U}- (\ln U -3 \ln V-2)V\p_{V}\Big),\\
X_{15}=\frac{U^{2} }{V^{2}}\Big( 4r\ln r\p_r+(6\ln V -2 \ln U+9)U\p_{U}+ (6\ln V -2 \ln U+7)V\p_{V}\Big).
\ea \end{equation}
\end{theo}

Although  the Lie algebra representation  (\ref{5-6}) is very
complicated, we were able to show that it is a classical Lie
algebra. In fact, Sophus Lie, who created foundations of the modern
Lie symmetry analysis, had proved that the simplest second-order ODE
$\ddot{y}(t)=0$ admits the eight-dimensional Lie algebra,  generated
by  the projective group of transformations on the plane,   and there
are no second-order ODEs possessing Lie algebras of a higher
dimensionality \cite{lie-1888} (see pages 225-226
therein)\footnote{\small To the best of our knowledge, this fact is
a little known; for example, it is not noted in the well-known books
devoted to the Lie symmetry analysis.}. Lie's result can be
generalised on two-component ODE systems.
 In particular, the simplest system of  two second-order ODEs (in what follows we use dots for time-derivatives)
\[ \ddot{y_1}(y_0)=0, \   \ddot{y_2}(y_0)=0, \   y_0\equiv t \]
admits the the fifteen-dimensional MAI
 generated by the Lie symmetries
 \begin{equation}\label{5-7}
 Y_{i}=\p_{y_i}, \ Y_{2}=y_j\p_{y_i},   \ Y_{i}=y_i(y_0\p_{y_0}+y_1\p_{y_1}+y_2\p_{y_2}), \quad i,j=0,1,2.
\end{equation}
It can  be noted that the four-dimensional subalgebra of (\ref{5-6})
\[ \langle X_{1}=\p_r, \ X_{2}=r\ln r\p_r,   \ X_{3}=r\frac{U^{2} }{V^{2}}\p_r, \ X_{4}=(\ln U -3 \ln V) r\p_r\rangle \]
has the same commutators as the four-dimensional Lie algebra
$\langle \p_{y_0}, \  y_0\p_{y_0}, \ y_1\p_{y_0}, \
y_2\p_{y_0},\rangle. $ Thus, we have two different representations
of the same Lie algebra and a local transformation must exist
transforming one algebra to another. In the case of Lie algebras of
low  dimensionality, the transformation can be found by direct
calculations. As a result, it can be shown that the relevant
transformation has the form
\begin{equation}\label{5-8}
 y_0=\ln r, \ y_1=\frac{U^2}{V^2},
 \ y_2= \ln{\frac{U}{V^3}}
 \Longleftrightarrow r= e^{y_0},
 \ U= y_1^{3/4} e^{-y_0/2}, \ V= y_1^{1/4} e^{-y_0/2}.
\end{equation}
Finally, using the above transformation, we have shown that all
other operators from (\ref{5-6}) are reducible to the relevant
operators (or their linear combinations) from the Lie algebra
(\ref{5-7}). For example,  $X_{5}$ and $X_{6}$  are transformed to
the operators $2y_1\p_{y_1}+ \p_{y_2}$ and $-2y_1\p_{y_1}-3
\p_{y_2}$, respectively. Thus, one concludes that the Lie symmetries
listed in (\ref{5-6}) and (\ref{5-7}) generate the same Lie algebra,
i.e. we have two representations (the terminology `realisations' is
used as well) of the same Lie algebra.

Generally speaking, there is no  guarantee that the above
observation means that transformation (\ref{5-8}) reduces the highly
nonlinear  system (\ref{5-5}) to the simplest linear system.
However, one may expect that  (\ref{5-5}) is integrable. It turns
out that this hypothesis is true! In fact, using formulae
(\ref{5-8}) and introducing two new functions $\Gamma_k(y_0), \
k=1,2$,  the nonlinear system (\ref{5-5}) is transformed to
\begin{equation}\label{5-9} \ba{l}\medskip
\dot{\Gamma_1}=3\dot{\Gamma_2}, \  \dot{\Gamma_2}+ (\Gamma_2-\Gamma_1)^2=0, \\

\Gamma_1= \frac{3}{4}y_1^{-1}\dot{y_1}- \frac{1}{2}y_2, \ \Gamma_2= \frac{1}{4}y_1^{-1}\dot{y_1}- \frac{1}{2}y_2
\ea \end{equation}
 (here dots denote differentiation with respect to $y_0$).

The first two ODEs in (\ref{5-9}) can be separately integrated and
having done this, the functions ${y_1}(y_0)$ and  ${y_2}(y_0)$  being
easily calculated from the last two ODEs. As a result, one arrives
at
\[ {y_1}(y_0)= \gamma_1(y_0+a_1), \quad {y_1}(y_0)= \gamma_2(y_0+a_2), \]
where $\gamma_k$ and $a_k \ (k=1,2)$ are arbitrary constants. Thus,
using (\ref{5-8}) and  new notations for the arbitrary constants, we
obtain the following exact solution of the nonlinear  system
(\ref{5-5})
\begin{equation}\label{5-10} U(r)=c_2 r^\gamma(c_1\ln (\kappa r))^{3/4}, \ V(r)=c_2 r^\gamma(c_1\ln (\kappa r))^{1/4}.
\end{equation}
Because this solution contains four arbitrary parameters ($\kappa \not=0$), we may claim that it is the general solution.
Interestingly, all solutions (\ref{5-10}) with $\gamma>0$
are bounded  at the point $r=0$ and this is important for their
prospective applications.

Notably, there are also two exceptional, but  trivial, solutions,
$U(r)=c_1, \ V(r)=c_2$ and $U(r)=c_1r, \ V(r)=c_2r$.

\subsection{Time-dependent exact  solutions}

 The five-dimensional MAI (\ref{5-4}) allows us to construct  a wide range of inequivalent reductions
 of the nonlinear system  (\ref{5-3}) to systems of ODEs. If fact, taking any linear combination of Lie symmetries from
 (\ref{5-4}),

\begin{equation}\label{5-11}
\alpha_1P_t +\alpha_2I_u +\alpha_3D_{00}+\alpha_4D_{11}+\alpha_5 X,
\end{equation}
a relevant ansatz can be constructed that  reduces  (\ref{5-3}) to a
system of ODEs. Having  obtained any solution of the ODE system, an
exact solution of  (\ref{5-3}) is  straightforwardly constructed. A
complete  description of all   inequivalent ans\"atze and
application for finding exact solutions lie  beyond scopes of  this
study. Here we analyse  two cases  leading to  non-trivial
time-dependent exact  solutions.

Let us assume that $\alpha_5=0$, i.e. the most complicated Lie
symmetry $X$ is excluded, then, making a simple  analysis, one
concludes that   the most interesting case occurs when
$\alpha_3\alpha_4\not=4$. In this case, one may set $\alpha_1=0$ and
$\alpha_3=1$ without losing  generality, therefore  the following
ansatz is  derived
\begin{equation}\label{5-12}
u(t,r)=t^{(b_0+b_1)/2} U(\omega), \quad v(t,r)=t^{b_1/2} V(\omega), \quad \omega=t^{-(1+b_1)/2}r,
\end{equation}
where new notations $b_0=\alpha_2 $ and $b_1=\alpha_4$ are introduced.
Substituting the above ansatz into  (\ref{5-3}) and making simplifications, one arrives at the following  ODE system
for the new variables $U(\omega)$ and $V(\omega)$:
\begin{equation}\label{5-13} \ba{l}
 V \left(U_{\omega,\omega}+\frac{U_{\omega}}{\omega}\right)+\frac{2V}{U}U_{\omega}^2+\frac{3U}{V}V_{\omega}^2- 6U_{\omega}V_{\omega}= \frac{2b_0+b_1}{2}\frac{U}{V}+\frac{1+b_1}{2}\frac{\omega}{V^2}(UV_\omega-2VU_\omega)\\
 \medskip\\
V \left(V_{\omega,\omega}+\frac{V_{\omega}}{\omega}\right)+\frac{V^{2} }{U^{2}}U_{\omega}^{2}-\frac{2 V}{U} U_{\omega}V_{\omega}= \frac{b_0+b_1}{2}-\frac{1+b_1}{2}\frac{\omega U_{\omega}}{U}.
 \ea \end{equation}

It turns out that the above system with $b_1=-1$ admits a highly non-trivial Lie symmetry. Indeed, (\ref{5-13}) with $b_1=-1$ is invariant with respect the three-dimensional MAI with the basic operators
\begin{equation}\label{5-14}  I_U=U\p_{U}, \ D_{11}= r\p_t+U\p_{U}+ V\p_{V}, \
X=\ln r \Big( r\p_r+U\p_{U}+ (1+\ln^{-1}r)V\p_{V}\Big)
 \end{equation}
($\omega$ with $b_1=-1$ is nothing else but $r$). There are
different techniques allowing to construct exact solutions of ODE
systems possessing non-trivial symmetry (see, e.g., Chapter 2 of
\cite{olv-93}). Notably, Sophus Lie worked out and applied the
theory for integration of ODEs using symmetries in his seminal work
\cite{lie-1888}. In principle, having the three-dimensional Lie
algebra (\ref{5-14}), system (\ref{5-13}) with $b_1=-1$ can be
transformed into four-component system of first-order ODEs, which
should be  reducible to a first-order ODE by a chain of
transformations. However, this technique is difficult   to realise
because typically  it is impossible to produce closed-form solutions
of the initial system from those of the first-order ODE obtained.
Moreover, even the single ODE can be non-integrable (see discussion
in \cite{lie-1888}).

Here we apply the standard approach, i.e. we  reduce the ODE  system (\ref{5-13}) with $b_1=-1$, that is
\begin{equation}\label{5-15} \ba{l}
 V \left(U_{r,r}+\frac{U_{r}}{r}\right)+\frac{2V}{U}U_{r}^2+\frac{3U}{V}V_{r}^2- 6U_{r}V_{r}=\frac{2b_0-1}{2}\frac{U}{V} \\
 \medskip\\
V \left(V_{r,r}+\frac{V_{r}}{r}\right)+\frac{V^{2} }{U^{2}}U_{r}^{2}-\frac{2 V}{U} U_{r}V_{r}= \frac{b_0-1}{2},
 \ea \end{equation}
 to algebraic equations using the three-dimensional Lie algebra (\ref{5-14}).
So,  taking any linear combination of Lie symmetries from (\ref{5-14}) and assuming non-zero coefficient for $X$ (otherwise the ansatz obtained  is rather trivial), one obtains the ansatz

  \begin{equation}\label{5-16} U(r)=c_1 r (\alpha_1+\ln r)^{\alpha_0}, \ V(r)=c_2 r (\alpha_1+\ln r)^{\alpha_1},
\end{equation}
where $c_1, c_2, \alpha_0$ and $\alpha_1$ are arbitrary constants at the moment.
Substituting the above ansatz into  (\ref{5-15}) and making simplifications, one arrives at the following  algebraic equations
\begin{equation}\label{5-17} \ba{l}
c_2^2(3\alpha_0^2 - 6\alpha_0\alpha_1+3\alpha_1^2-\alpha_0)(\alpha_1+\ln r)^{2\alpha_1-2} =\frac{2b_0-1}{2}, \\
 \medskip\\
c_2^2(\alpha_0^2 - 2\alpha_0\alpha_1+\alpha_1^2 - \alpha_1)(\alpha_1+\ln r)^{2\alpha_1-2}= \frac{b_0-1}{2}.
 \ea \end{equation}
Now we observe that $\alpha_1=1$ is the only possibility to satisfy equations in (\ref{5-17}).
As a result, the exact solution
\[  U(r)=c_1 r (1+\ln r)^{\alpha_0}, \ V(r)=\pm c_2 r (\alpha_1+\ln r), \]
follows, where $c_2^2=\frac{1}{2(\alpha_0^2 - 3\alpha_0+3)}, \
b_0=1+ \frac{\alpha_0^2 - 2\alpha_0}{\alpha_0^2 - 3\alpha_0+3}. $

Finally, taking into account ansatz (\ref{5-12}), the non-stationary exact solution
\begin{equation}\label{5-18}
u(t,r)=c_1t^{(b_0-1)/2}r (1+\ln r)^{\alpha_0} , \quad v(t,r)=\pm c_2t^{-1/2} r (1+\ln r),
\end{equation}
of the nonlinear system (\ref{5-3}) is constructed.
Notably, similarly to the stationary solutions (\ref{5-10}), those of the form   (\ref{5-18})
 are bounded  at the point $r=0$ and this is important
for their prospective applications. If there is  need  to avoid singularity at $t=0$ then the time shift $t \to t+t_0, \ t_0>0$ leads to the continuous solution  for all $t\geq 0$.

The ODE system  (\ref{5-13})  with arbitrary parameters $b_0$ and
$b_1$ admits the two-dimensional Lie algebra with the basic
operators $I_U$ and $D_{11}$ from (\ref{5-14}). In this case, the
ansatz takes the form
  \begin{equation}\label{5-19} U(\omega)=c_1 \omega ^{\alpha}, \ V(r)=c_2 \omega
\end{equation}
and, after rather simple calculations, one derives  time-dependent
solutions of  the nonlinear system (\ref{5-3})
\begin{equation}\label{5-20}
u(t,r)=c_1t^{1/2}r ^{\alpha} , \quad v(t,r)=c_2t^{-1/2} r ,  \quad  c_2=\pm \frac{1}{\sqrt{2}(\alpha-1)}, \alpha\not=1.
\end{equation}

Let us demonstrate    a successful  application  of the Lie symmetry
(\ref{5-11})  with $\alpha_5 \not=0$ for finding exact solutions.
Setting  $\alpha_1=\alpha_3=0$ in (\ref{5-11}) , we obtain the
ansatz
\begin{equation}\label{5-21}
u(t,r)= r (b+\ln r)^{\alpha_0} U(t) , \quad v(t,r)= r (b+\ln r)V(t),
\end{equation}
where $b$  and $\alpha_0$ are arbitrary parameters  at the moment.
Substituting the above ansatz into  (\ref{5-3}), one  derives the ODE system
\begin{equation}\label{5-22} \ba{l}
 \dot U = c_1^2(\alpha_0^2-2\alpha_0)UV^2, \\
 \medskip\\
  \dot V = -c_2^2(\alpha_0^2-3\alpha_0+3)V^3
 \ea \end{equation}
 (here dots denote differentiation with respect to $t$).
 Integrating the above first-order ODEs  and using  (\ref{5-21}), one readily obtains
 the exact solution
 \begin{equation}\label{5-23}
u(t,r)=c_1(t+t_0)^{\kappa}r (b+\ln r)^{\alpha_0} , \quad v(t,r)=\pm c_2(t+t_0)^{-1/2} r (b+\ln r),
\end{equation}
where $c_1, \ \alpha_0, \  t_0$ and $b$ are arbitrary parameters, while  $c_2^2=\frac{1}{2(\alpha_0^2 - 3\alpha_0+3)}, \  \kappa= \frac{\alpha_0^2 - 3\alpha_0}{2(\alpha_0^2 - 3\alpha_0+3)}$.

  It can be noted that  (\ref{5-23})  is  a  four-parameter  family of exact solutions that contains the exact solutions (\ref{5-18})
   as a particular case.
  On the other hand, one observes that two essentially different ans\"atze,
   (\ref{5-12})  and (\ref{5-21}),  lead  to  exact solutions with the same structure.
  So, this is another confirmation of  the well-known statement that  different symmetries  may
  lead to the same/similar exact solutions (see an extensive discussion on this matter
  in Chapter 4  of \cite{ch-se-pl-book}).

  In conclusion of this section, we highlight the following observation.
  Taking into account that  the Lie symmetry of the two-component system (\ref{4-3}) and  the four-component system (\ref{4-3}) and (\ref{4-5})
 is (surprisingly!)  the same,     a natural   question arises: can such  exact solutions of the two-component system
(\ref{4-3})
be constructed that are also  solutions of the four-component system
(\ref{4-3}) and (\ref{4-5}) (the latter is equivalent to
(\ref{4-2})) ? The general theory of PDEs implies that
in general this will not be the case otherwise the two-component
system would be  equivalent to the four-component one. We checked
the solutions derived above and found that they do not satisfy the
overdetermined system (\ref{4-3}) and (\ref{4-5}). However, it is
noteworthy that some of them {\it do } satisfy the latter under additional
parameter restrictions: for example, the stationary solution
(\ref{5-10}) is an exact solution of (\ref{4-3}) and (\ref{4-5})
provided $\gamma=1$.

\section{\bf Conclusions  }\label{sec:6}

The  results obtained  here can be summarised as follows.
\begin{enumerate}
            \item All possible two-components evolutions systems of (1+2)-dimensional second-order PDEs admitting the
            infinite-dimensional Lie algebra generated by the Lie symmetry   (\ref{1-1})  have been  described (see Theorem 1).
            \item It was shown that a natural generalisation of the  Lie symmetry (\ref{1-1})
            to the  higher-dimensional case does not lead to a relevant generalisation of Theorem 1
            because the infinite-dimensional symmetry is broken.
            \item  The recently derived  system with two space variables (see (\ref{4-3})-(\ref{4-4})), which is related to Ricci flows
\cite{lo-dimas-bo-2023},   was  identified   as a very particular case of Theorem 1.
\item  It  was  shown that an extension of the above  two-component system  to an overdetermined   four-component system
 (as suggested in \cite{lo-dimas-bo-2023}) does not change the Lie
 symmetry.
However, a reduced  (three-component) system  does not possess an
infinite-dimensional  Lie symmetry.
 The conditional symmetry concept \cite{fss}  was  used in
order to clarify this phenomenon.
\item All possible  stationary solutions in the radially symmetric case  were  constructed
using the surprisingly rich Lie algebra (\ref{5-6}) of the reduced
system of ODEs.  Moreover, it  was  proved that the Lie
algebra (\ref{5-6}) is reducible to the fifteen-dimensional algebra
of the simplest system of two second-order ODEs.
\item Several time-dependent exact solution in the radially symmetric case were  constructed as
well.  It  was  shown that the solutions obtained are
bounded and smooth provided the arbitrary parameters 
specified suitably. 
        \end{enumerate}

       This work could be continued in different directions.
        Firstly, a possible   generalisation
            of  the results obtained  to systems of $(1+n)$-dimensional second-order PDEs ( $n\geqq 3$)
            is an open problem: a correct generalisation of  (\ref{1-1}) leading to an infinite-dimensional Lie symmetry
             would need to  be constructed. Secondly, exact solutions of  the nonlinear
system (\ref{4-3}) with $m\not=2$ should be constructed because the
case  $m=2$ is special (similarly to $m=1$  being a special case for
the relevant system with a single space variable  \cite{ch-ki-24}).
Thirdly,  more complicated ans\"atze (compared to (\ref{5-2})) could
be constructed using the infinite-dimensional symmetry
in the search for exact solutions.

\section{Acknowledgement}
R.Ch. acknowledges that this research was funded by the British
Academy's Researchers at Risk Fellowships Programme. J.R.K.
gratefully acknowledges a Royal Society Leverhulme Trust Senior Fellowship.

\end{document}